\documentclass[pra,floatfix,amsmath,superscriptaddress,onecolumn,nofootinbib]{revtex4-2}
\usepackage[T1]{fontenc} 
\usepackage{amssymb}
\usepackage{graphicx}
\usepackage{graphics}
\usepackage{amsmath}
\usepackage{amsthm}
\usepackage{color}
\usepackage{braket}
\usepackage[utf8]{inputenc}

\usepackage{dcolumn}
\usepackage{bm}
\usepackage{soul}

\newcommand{\bea}{\begin{eqnarray}}
\newcommand{\eea}{\end{eqnarray}}

\def\bi{\begin{itemize}}
\def\ei{\end{itemize}}
\def\bc{\begin{center}}
\def\ec{\end{center}}

\def\C{\hbox{$\mit I$\kern-.7em$\mit C$}}
\def\R{\hbox{$\mit I$\kern-.6em$\mit R$}}

\def\ket#1{|#1\rangle}

\def\ket#1{\left| #1\right>}

\def\c{\mathcal{C}}

\def\a{\ket{a}}
\def\b{\ket{b}}
\def\c{\ket{c}}
\def\d{\ket{d}}

\newtheorem{theorem}{Theorem}
\newtheorem{corollary}[theorem]{Corollary}
\newtheorem{lemma}[theorem]{Lemma}
\newtheorem{proposition}[theorem]{Proposition}
\newtheorem{definition}[theorem]{Definition}
\newtheorem{remark}[theorem]{Remark}
\newtheorem{example}{Example}

\begin{document}

\title{Entanglement gain in supercatalytic state transformations}
\author{Guillermo Díez-Pastor}\email{g.dpastor@upm.es}
\affiliation{Departamento de Matem\'atica Aplicada a la Ingeniería Industrial, Universidad Polit\'ecnica de Madrid, E-28006 Madrid, Spain}
\author{Julio I. de Vicente}\email{jdvicent@math.uc3m.es}
\affiliation{Departamento de Matem\'aticas, Universidad Carlos III de
Madrid, E-28911, Legan\'es (Madrid), Spain}
\affiliation{Instituto de Ciencias Matem\'aticas (ICMAT), E-28049 Madrid, Spain}

\begin{abstract}
Catalysis refers to the possibility of performing an otherwise impossible local state transformation by sharing an additional state, i.e.\ a catalyst, which is returned at the end of the protocol. There is a stronger version, known as supercatalysis, in which the borrowed catalyst is returned in an enhanced form, i.e.\ more entangled. However, this phenomenon has remained little explored. In this work we introduce the supercatalytic entanglement gain as a figure of merit taking values in $[0,1]$ that quantifies the performance of the protocol (with 0 corresponding to the standard case of catalysis and 1 representing the maximal possible gain) and we study in which cases it can be greater than zero and which strategies can maximize it. While it turns out that every catalytic transformation can be implemented in a supercatalytic fashion with entanglement gain equal to 1 if the state that is borrowed is chosen appropriately, other choices can make the gain strictly less than 1 and even 0. In fact, we prove that a large class of catalytic transformations are not fully supercatalyzable, i.e.\ there is at least one choice of catalyst for which the entanglement gain vanishes. On the other hand, the construction that shows that supercatalysis is always possible with maximal gain uses artificially highly entangled catalysts. For this reason, we also study minimal supercatalysis, where the entanglement content of the borrowed state is constrained in a precise and natural way. While we consider a scenario where we prove it is impossible to have entanglement gain equal to 1 in this case, we show that there exist minimal supercatalytic transformations with gain as close to 1 as desired. We also explore several examples and observe that, although choosing a catalyst with the least possible entanglement is often an optimal strategy for minimal supercatalysis, this is not necessarily always the case.
\end{abstract}

\maketitle

\section{Introduction}\label{sec1}

Entanglement is a purely quantum phenomenon with no classical analogue that is at the heart of the foundations of quantum mechanics. With the advent of quantum technologies that promise to drastically overcome the performance of their classical counterparts, in the last three decades entanglement theory has been developed with the aim to rigorously characterize its capabilities and limitations as a resource (see the reviews \cite{PlenioVirmani, losHorodecki}). The paradigm of state transformations under local operations and classical communication (LOCC) provides the cornerstone of this theory. First, this underpins the basic protocols with which this resource can be manipulated by spatially separated parties for information processing tasks. Second, and from a more abstract point of view, this makes it possible to operationally define an order in the set of entangled states and provides the basic axioms for the construction of entanglement measures. A crucial result in this context is Nielsen's theorem \cite{Nielsen}, which completely characterizes which LOCC transformations are possible between arbitrary pure bipartite states in terms of an easy-to-check majorization-type condition \cite{MarOlk}. Thanks to this theorem, Jonathan and Plenio identified in \cite{JP} the striking effect of entanglement catalysis, which plays a key role in the structure and applications of entanglement theory. Entanglement catalysis stands for the fact that there exist pairs of bipartite states such that one cannot be transformed into the other by LOCC but for which the conversion becomes possible if the parties share an additional entangled state that can be returned intact at the end of the protocol. This auxiliary state acts much like a catalyst in a chemical reaction. Its presence makes a transformation possible and, since it is not consumed, it can be reused.

The discovery of entanglement catalysis in \cite{JP} spurred considerable further research, and the amount of literature on the topic is vast. For instance, Refs.\ \cite{turgut,klimesh} characterize the pairs of states for which catalytic transformations are possible and \cite{DaftuarKlimesh,XiaoRunMin,gourcatalyst1,gourcatalyst2} investigate which states can act as catalysts. Furthermore, different variations and generalizations have been put forward such as mixed-state catalysis \cite{mixedcatalysis}, multipartite catalysis \cite{multipartitecatalysis}, probabilistic catalysis \cite{JP}, mutual catalysis \cite{FWX}, self-catalysis \cite{selfcatalysis}, approximate catalysis \cite{approximatecatalysis}, or correlated catalysis \cite{correlatedcatalysis}. Additionally, this phenomenon has been exhaustively studied beyond entanglement theory in the context of other quantum resource theories \cite{qrts}. Catalysis is relevant in the study of basic quantum information protocols like entanglement embezzlement \cite{embezzlement}, quantum state merging \cite{correlatedcatalysis}, entanglement distillation \cite{correlatedcatalysis}, quantum teleportation \cite{teleportation} and decoupling \cite{decoupling}. It can be used to enhance the capacity of noisy quantum channels \cite{catalysischannels} and as a basic ingredient in cryptographic protocols \cite{crypto1,crypto2}, and it plays a prominent role in the foundations of quantum thermodynamics \cite{approximatecatalysis,thermal}. For a more detailed presentation and a more complete list of references, we direct our readers to the recent and exhaustive review articles \cite{Datta} and \cite{Lipka}.

An extremely natural generalization of entanglement catalysis is entanglement supercatalysis, which is the subject of this work. Herein, instead of obtaining in the auxiliary system the same state at the end of the transformation protocol, one asks that the borrowed state is returned in an objectively improved form. In fact, in Jonathan and Plenio's seminal article on catalysis, the auxiliary system is referred to as an ``entanglement banker'' that is called Scrooge. Why would Scrooge settle with getting back only what he lent when the same transformation can be carried out in the main system allowing him to get an interest back in return? This possibility comes with obvious advantages because the borrowed state can not only be reused as in the standard catalytic scenario, but also generate an entanglement gain, thus augmenting its capacity to facilitate subsequent tasks. This problem is related to the notion of partial entanglement recovery \cite{Mori,BanRoyVatan,DuanFengYing}, where, contrary to supercatalysis, one considers a transformation that is indeed possible by LOCC alone and introduces an auxiliary system to try to recover part of the entanglement lost in the main system. Supercatalysis was introduced and its existence proved in \cite{BandRoy} (other examples can be found e.g.\ in \cite{FWX} as instances of mutual catalysis). In addition to this, this work provides an (inefficient) algorithm to look for catalysts and supercatalysts and provides instances of transformations where the use of certain catalysts does not allow any entanglement gain. This still leaves many questions unanswered, and, quoting \cite{Lipka} (see the last paragraph in Sec.\ III.A), supercatalysis has remained little studied ever since and it provides an interesting avenue for further research.    

In this work we take the point of view of Scrooge and seek for the best state that we can lend in order to maximize our entanglement profit in the returned state while at the same time allowing a given and otherwise impossible LOCC state transformation in the main system. For this, one needs a clear definition of what enhancing the entanglement in the auxiliary system means. Reference \cite{BandRoy} requires in principle that the entanglement entropy of the auxiliary system increases. Although the entanglement entropy is the most paradigmatic bipartite entanglement measure and has a clear operational interpretation in terms of entanglement distillation and concentration \cite{entanglemententropy}, it may increase while some other also relevant entanglement measure decreases. To make this notion universal and not subject to a particular task that a particular entanglement measure might be attached to, and in the same vein as the approach taken by \cite{DuanFengYing} in the context of partial entanglement recovery, we require this enhancement to be independent of the entanglement measure used. That is, we demand that the returned state can be converted by LOCC into the borrowed state (which is equivalent to the returned state having no less entanglement than the borrowed state according to \emph{any} entanglement measure). As already observed in \cite{BandRoy}, this implies that the transformation in the main system must be possible catalytically and that both the borrowed and the returned states must be themselves catalysts. In order to quantify the profit made by Scrooge we are, however, bound to a choice of entanglement measure and we use then the entanglement entropy. This allows us to introduce the entanglement gain relative to the transformation that wants to be accomplished, which is a function of the catalyst used. This parameter takes values in $[0,1]$, where 0 corresponds to no increase of entanglement (i.e.\ catalysis) and 1 to the maximum possible gain. We then study when transformations at entanglement gain strictly larger than 0 are possible and what the optimal strategies are for Scrooge to maximize this figure of merit.

The article is organized as follows. In Sec.\ II we provide all necessary background and definitions. In Sec.\ III.A we provide a very simple construction that shows that every catalytic transformation can be implemented supercatalytically at maximum entanglement gain if the catalyst is chosen wisely. The choice of catalyst that is lent is crucial, other options can lead to a smaller entanglement gain, even to the point of making it vanish, an example of which was already provided in \cite{BandRoy}, as already mentioned above. This leads us to introduce the notion of fully supercatalyzable transformation, which amounts to a state transformation that can be implemented supercatalytically with non-null entanglement gain for any choice of catalyst. In Sec.\ III.B we give several no-go results that prove that full supercatalysis is impossible for a large class of transformations, including in particular those for which the input and output states have the same Schmidt rank and those that can be implemented with a catalyst with minimal Schmidt rank (i.e.\ equal to 2). The construction of Sec.\ III.A is somewhat artificial and uses catalysts of unnecessarily large Schmidt rank. For this reason we consider in Sec.\ IV minimal supercatalysis, in which the lent catalyst is bound to the minimal necessary Schmidt rank and, hence, the observation of Sec.\ III.A no longer needs to hold. In fact, we show herein that in the simplest scenario where the input and output states have Schmidt rank less than or equal to 4 and have a catalyst with Schmidt rank equal to 2, then minimal supercatalysis at entanglement gain exactly equal to 1 is impossible. However, we prove that there exists a family of such supercatalytic transformations with entanglement gain as close to 1 as desired. This shows that a non-trivial universal bound (i.e.\ independent of the states involved) on the entanglement gain does not exist even in the minimal scenario. However, we do provide an upper bound dependent on the input and output states and the catalyst used and test its performance with several examples. These also show that lending the least entangled catalyst (what we call the miserly strategy) is often optimal in order to maximize the entanglement gain. However, we also find instances where this is not the case. These results show that Scrooge's optimal strategy has a complex dependence on the input and output states and suggest that determining it in general might be a formidable problem.  

\section{Preliminaries}\label{sec2}

\subsection{Background}

In this article we will always consider finite-dimensional bipartite quantum systems described by a pure state. That is, the system has a Hilbert space $\mathcal{H}$ associated to it such that $\dim\mathcal{H}<\infty$ and $\mathcal{H}=\mathcal{H}_A\otimes\mathcal{H}_B$, where $\mathcal{H}_A$ is the Hilbert space associated to one party, referred to as Alice, and $\mathcal{H}_B$ is the one associated to the other, referred to as Bob. The set of all pure states is then given by $S(\mathcal{H})=\{\ket{a} \in \mathcal{H} : ||\ket{a}||=1 \}$. LOCC transformations correspond to a particular class of completely-positive and trace-preserving maps that capture that only classical but not quantum information is exchanged between the parties (for a precise definition see e.g.\ \cite{locc1,locc2}). Local unitary transformations are reversible LOCC transformations and in the study of LOCC convertibility it is convenient to take any given state $\ket{a}\in S(\mathcal{H})$ to be in Schmidt form, i.e.\
\begin{equation}\label{schmidt}
    \ket{a} = \sum_{i=1}^d \sqrt{a_i} \ket{e_i} \ket{f_i},
\end{equation} 
where $d=\min\{\dim\mathcal{H}_A,\dim\mathcal{H}_B\}$,  $\{\ket{e_i}\}$ and $\{\ket{f_i}\}$ are orthonormal sets in respectively $\mathcal{H}_A$ and $\mathcal{H}_B$, and $\{a_i\}$, the so-called Schmidt coefficients, satisfy $a_i\geq0$ $\forall i$ and $\sum_ia_i=1$. Note that the Schmidt decomposition implies that without loss of generality we can take $\dim\mathcal{H}_A=\dim\mathcal{H}_B=d$. The set of Schmidt coefficients is obviously invariant under local unitary transformations and contains all the information about LOCC convertibility. We will denote the vector of Schmidt coefficients of $\ket{a}\in S(\mathcal{H})$ by $a$, which is an element of the polytope $\Delta_d = \{ v \in \mathbb{R}^d : v_i \geq 0 \ \forall i$, $\sum_i v_i = 1\}$ of probability distributions over $d$ outcomes.

Given two states $\ket{a},\ket{b}\in S(\mathcal{H})$, if there exists an LOCC protocol that transforms $\ket{a}$ into $\ket{b}$, we will write $\ket{a}\to\ket{b}$. The reason why we do not need here a proper definition of the class of LOCC maps is that Nielsen's theorem \cite{Nielsen} completely characterizes these transformations in our scenario with a clear-cut majorization condition formulated over the Schmidt coefficients \cite{MarOlk}. Namely, given $\ket{a},\ket{b}\in S(\mathcal{H})$ it holds that $\ket{a}\to\ket{b}$ if and only if (iff) $b\succ a$, i.e.\ $b$ majorizes $a$ \cite{MarOlk}. This means that $f_k(a)\leq f_k(b)$ for $k=1,2,\ldots,d$, where here and in the rest of the article, for any $v\in\Delta_d$, $f_k(v)$ stands for the sum of the $k$ largest entries of $v$. Entanglement measures for pure states $M:S(\mathcal{H})\to[0,\infty)$ must satisfy $M(\ket{a})\geq M(\ket{b})$ if $\ket{a}\to\ket{b}$ and they must vanish on non-entangled states (i.e.\ those that have a unique non-zero Schmidt coefficient). Among other things, Nielsen's theorem tells us that entanglement measures are in one-to-one correspondence with Schur-concave functions on $\Delta_d$ \cite{MarOlk}, i.e.\ an entanglement measure can be written as $M(\ket{a})=g(a)$ where $g:\Delta_d\to[0,\infty)$ is any (properly rescaled) function satisfying $g(a)\geq g(b)$ whenever $b\succ a$. As explained in the introduction, given its interpretation in the asymptotic scenario, the most widely used entanglement measure is the entanglement entropy, $E$, which corresponds to the choice of the Shannon entropy for the function $g$, that is
\begin{equation}
E(\ket{a})=-\sum_{i=1}^da_i\log a_i.
\end{equation}
It is relevant to point out that the Shannon entropy is not only Schur-concave but strictly Schur-concave \cite{MarOlk}. This means that if $\ket{a}\to\ket{b}$, then it must hold that $E(\ket{a})>E(\ket{b})$ unless $\ket{a}$ and $\ket{b}$ are local unitarily equivalent (i.e.\ unless $f_k(a)=f_k(b)$ $\forall k$, which is the same as saying that the Schmidt coefficients are equal up to reordering). In addition to the entanglement entropy, there are many other entanglement measures of relevance. In this work we will only need to consider the Schmidt rank \cite{schmidt}, which is given by the number of non-zero entries of the corresponding Schmidt vector and will be denoted by $SR(|a\rangle)$.  

In catalytic transformations, in addition to the main system we have an auxiliary system for the catalyst. This means that, while the global Hilbert space remains split by two parties, i.e.\ $\mathcal{H}=\mathcal{H}_A\otimes\mathcal{H}_B$, the local Hilbert spaces further decompose as $\mathcal{H}_{A} = \mathcal{H}_{A_1} \otimes \mathcal{H}_{A_2}$ and $\mathcal{H}_{B} = \mathcal{H}_{B_1} \otimes \mathcal{H}_{B_2}$. Thus, $\mathcal{H}_1= \mathcal{H}_{A_1} \otimes \mathcal{H}_{B_1}$ corresponds to the main system, while $\mathcal{H}_{2} = \mathcal{H}_{A_2} \otimes \mathcal{H}_{B_2}$ corresponds to the auxiliary system (see Fig.\ \ref{fig:dibujin}). Given $\ket{a},\ket{b}\in S(\mathcal{H}_1)$ and $\ket{c}\in S(\mathcal{H}_2)$, we say that $\ket{a}$ can be catalytically converted to $\ket{b}$ and we denote it by $\ket{a} \rightarrow_c \ket{b}$ if $\ket{a}\ket{c}\to\ket{b}\ket{c}$. The state $\ket{c}$ is referred to as catalyst. The surprising observation by Jonathan and Plenio is that it can hold that $\ket{a}\nrightarrow\ket{b}$ while $\ket{a} \rightarrow_c \ket{b}$. Mathematically, the existence of catalysis is a relatively simple consequence of majorization. The Schmidt vector of the state $\ket{a}\ket{c}\in S(\mathcal{H})$ happens to be given by the Kronecker product of the Schmidt vectors of $\ket{a}$ and $\ket{c}$, i.e.\ $a\otimes c$. Thus, the aforementioned phenomenon amounts to finding $a,b,c\in\Delta_d$ such that $b\nsucc a$ but $b\otimes c\succ a\otimes c$. Plenty of such examples can be found in the references given in the introductory section. 

\begin{figure}
\includegraphics[width=0.4\textwidth]{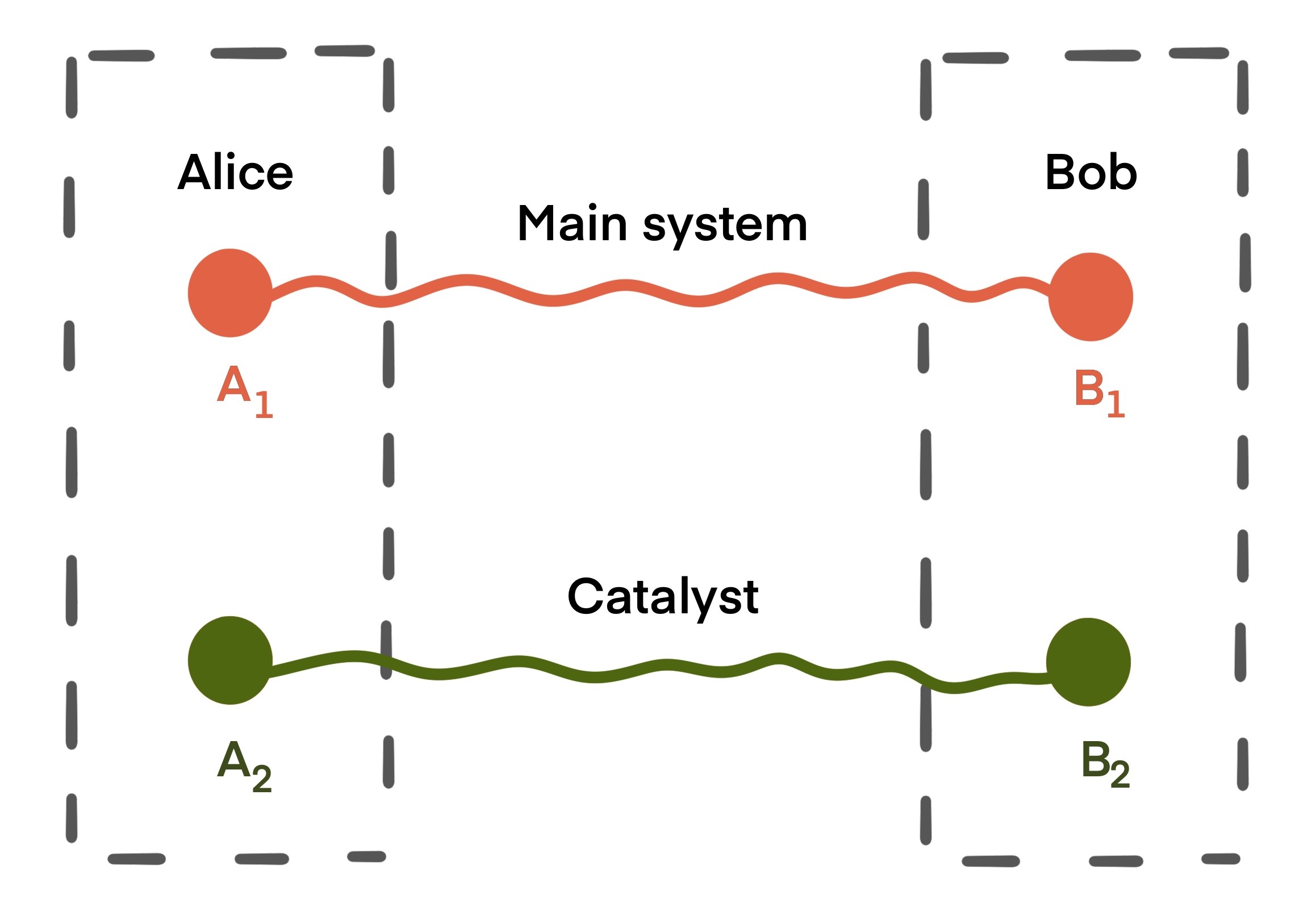}
\caption{\label{fig:dibujin} Alice and Bob share two entangled states: the main system, on which they want to perform a certain transformation, and an auxiliary system that acts as a catalyst for the transformation.}
\end{figure}

A comment is now in order for the notation that we will follow for the rest of the paper. Note that by reordering the elements of the Schmidt bases the Schmidt decomposition of any state as given in Eq.\ (\ref{schmidt}) can be taken without loss of generality such that the Schmidt coefficients are ordered non-increasingly. However, we have not imposed from the beginning that all Schmidt vectors have their entries ordered, as $a\otimes c$ is not ordered even if $a$ and $c$ are. From now on, we will consider without loss of generality that all Schmidt vectors of states in $S(\mathcal{H}_1)$ or in $S(\mathcal{H}_2)$ (i.e.\ of the main or the auxiliary system) are ordered. However, the joint Schmidt vector of the system and the catalyst, as given by $a\otimes c$, cannot be taken to be ordered, and, for this, we will write $(a\otimes c)^\downarrow$. Whenever we want to constrain the polytope of all Schmidt vectors to be ordered we use $\Delta'$, i.e.\
\begin{equation}
\Delta'_d = \{ v \in \mathbb{R}^d : v_i \geq 0\textrm{ and }v_{i} \ge v_{i+1}\, \forall i,\, \sum_i v_i = 1\}.
\end{equation}

Given a pair $\ket{a},\ket{b}\in S(\mathcal{H}_1)$ such that $\ket{a} \rightarrow_c \ket{b}$, we define
\begin{equation}
C(\ket{a},\ket{b})=\{\ket{c}\in S(\mathcal{H}_2):\ket{a}\ket{c}\to\ket{b}\ket{c}\}.
\end{equation}
That is, $C(\ket{a},\ket{b})$ is the set of all states that can act as catalysts in a given catalytic transformation. For $r\in\mathbb{N}$, we will also consider the sets
\begin{equation}
C_r(\ket{a},\ket{b})=\{\ket{c}\in S(\mathcal{H}_2):SR(\ket{c})\leq r, \ket{a}\ket{c}\to\ket{b}\ket{c}\}
\end{equation}
of catalysts for a given catalytic transformation with Schmidt rank not larger than $r$. Note that, if $\ket{a}\nrightarrow\ket{b}$ but $\ket{a} \rightarrow_c \ket{b}$, then all catalysts for this transformation must be entangled, i.e.\ $C_1(\ket{a},\ket{b})=\emptyset$. As already mentioned in the introductory section, there exist several works that investigate which states can act as catalysts for a given nontrivial catalytic transformation $\ket{a} \rightarrow_c \ket{b}$ (i.e.\ such that $\ket{a}\nrightarrow\ket{b}$). For instance, it is proved in \cite{JP} that the maximally entangled state of Schmidt rank equal to $r$, i.e.\ with Schmidt vector given by $(1/r,1/r,\ldots,1/r)$, can never act as catalyst. However, a full characterization of catalysts is only known in the simplest scenario in which the Schmidt ranks of the states involved take their minimal possible values \cite{XiaoRunMin}. This result will be used later and we state it in the following in precise terms. Before that, let us mention that it was proved in \cite{JP} that if $\ket{a}\nrightarrow\ket{b}$ but $\ket{a} \rightarrow_c \ket{b}$ with $SR(|a\rangle),SR(|b\rangle)\leq4$, then it must hold that
\begin{equation}\label{eq:thcatalyst0}
f_1(a)\leq f_1(b),\quad f_2(a)>f_2(b),\quad f_3(a)\leq f_3(b).
\end{equation}
\begin{theorem}[\cite{XiaoRunMin}]
\label{th:catalyst}
Let $\ket{a},\ket{b}\in S(\mathcal{H}_1)$ be such that $SR(|a\rangle),SR(|b\rangle)\leq4$ and $\ket{a}\nrightarrow\ket{b}$. Then, there exists a catalyst $\ket{c}\in C_2(\ket{a},\ket{b})$ with Schmidt vector $c = (x, 1-x)$ for $x \in [0.5, 1]$ iff Eq.~(\ref{eq:thcatalyst0}) holds and
\begin{equation}
x_{\min}(a,b):=\max \Bigl\{ \frac{a_1 + a_2 - b_1}{b_2 + b_3} \ , 1-\frac{a_4-b_4}{b_3-a_3}\Bigl\}\leq x\leq x_{\max}(a,b):=\min \Bigl\{ \frac{b_1}{a_1 + a_2} \ , \frac{b_1-a_1}{a_2 - b_2}  \ , 1-\frac{b_4}{a_3+a_4}\Bigl\}.
\end{equation}
\end{theorem}  
It was shown in \cite{JP} that, in a nontrivial catalytic transformation, it cannot hold that $\a$ and $\b$ have both Schmidt ranks equal to two or three. As a consequence, since the Schmidt rank cannot increase under LOCC, the only possible nontrivial catalytic transformations with $SR(|a\rangle),SR(|b\rangle)\leq4$ are those where $SR(|a\rangle)=4$ and either $SR(|b\rangle)=3$ or $SR(|b\rangle)=4$. 

\subsection{Supercatalysis}

We begin with a proper definition of supercatalysis as explained and motivated in Sec.\ \ref{sec1}.

\begin{definition}\label{def:supercatalysis}
Given $\ket{a},\ket{b}\in S(\mathcal{H}_1)$ such that $\ket{a}\nrightarrow\ket{b}$, we say that there is a \emph{supercatalytic transformation} from $\ket{a}$ into $\ket{b}$ if there exist $\ket{c},\ket{d}\in S(\mathcal{H}_2)$ such that $c\neq d$ and $\ket{a}\ket{c}\to\ket{b}\ket{d}$ and $\ket{d}\to\ket{c}$.
\end{definition}

Note that the condition $c\neq d$ is imposed in order to exclude the standard case of catalytic transformations. As already pointed out in \cite{BandRoy}, the very definition of supercatalysis that we use imposes certain relevant structure on the states involved in supercatalysis. We state again here this important property and provide as well its proof for the sake of completeness.

\begin{proposition}[\cite{BandRoy}]\label{mainobs}
If $\ket{a}\in S(\mathcal{H}_1)$ can be supercatalytically transformed into $\ket{b}\in S(\mathcal{H}_1)$ with states $\ket{c},\ket{d}\in S(\mathcal{H}_2)$ as given in Definition \ref{def:supercatalysis}, then $\ket{a} \rightarrow_c \ket{b}$ and both $\ket{c},\ket{d}\in C(\ket{a},\ket{b})$.
\end{proposition}
\begin{proof}
Simply note that Definition \ref{def:supercatalysis} implies the following chain of LOCC transformations:
\begin{equation}
\ket{a}\ket{d} \rightarrow \ket{a}\ket{c} \rightarrow \ket{b}\ket{d} \rightarrow \ket{b}\ket{c},
\end{equation}
where the first and the last transformation follow from the fact that $\ket{d}\to\ket{c}$. Thus, it must hold that $\ket{a}\ket{d} \rightarrow \ket{b}\ket{d}$ and $\ket{a}\ket{c} \rightarrow \ket{b}\ket{c}$.
\end{proof} 

This observation tells us that supercatalytic transformations are only possible among catalytic transformations and that to search for Scrooge's best strategies regarding which state to lend and which state to ask for in return, one only needs to explore the structure of the set of all catalysts for a given transformation. However, for any given catalytic transformation, it is in principle not clear which choices of catalysts make supercatalysis possible, if any at all. For this reason, as a first step in the analysis of this phenomenon we introduce the following definitions.

\begin{definition}\label{supercatalyzable}
Given $\ket{a},\ket{b}\in S(\mathcal{H}_1)$ such that $\ket{a}\nrightarrow\ket{b}$ but $\ket{a} \rightarrow_c \ket{b}$, the catalytic transformation is said to be \emph{supercatalyzable} if there exists at least one choice of catalysts $\ket{c},\ket{d}\in C(\ket{a},\ket{b})$ such that $c\neq d$ and $\ket{a}\ket{c}\to\ket{b}\ket{d}$ and $\ket{d}\to\ket{c}$.
\end{definition} 

\begin{definition}\label{fullysupercatalyzable}
Given $\ket{a},\ket{b}\in S(\mathcal{H}_1)$ such that $\ket{a}\nrightarrow\ket{b}$ but $\ket{a} \rightarrow_c \ket{b}$, the catalytic transformation is said to be \emph{fully supercatalyzable} if for every catalyst $\ket{c}\in C(\ket{a},\ket{b})$, there exists another catalyst $\ket{d}\in C(\ket{a},\ket{b})$ such that $c\neq d$ so that $\ket{a}\ket{c}\to\ket{b}\ket{d}$ and $\ket{d}\to\ket{c}$ hold.
\end{definition}

Thus, a catalytic transformation that is supercatalyzable means that there exists at least one choice of loan state for Scrooge that allows him to upgrade it to supercatalysis and make some entanglement profit, while one that is fully supercatalyzable has the property that every potential loan state (i.e.\ every catalyst) makes in fact profit possible.

If a supercatalyzable catalytic transformation is fully supercatalyzable or, at least, it can be upgraded to supercatalysis by different choices of catalysts as loan states, then it is natural to study which choice performs better. To quantify Scrooge's profit, as explained in the introduction, we introduce the supercatalytic entanglement gain. For this, it is important to observe that the entanglement entropy is additive, i.e.\ for any $\ket{a}\in S(\mathcal{H}_1)$ and any $\ket{c}\in S(\mathcal{H}_2)$ it holds that $E(\ket{a}\ket{c})=E(\ket{a})+E(\ket{c})$. This, together with the aforementioned fact that Shannon entropy is strictly Schur-concave, implies that whenever $\ket{a}\nrightarrow\ket{b}$ but $\ket{a} \rightarrow_c \ket{b}$ it must hold that $E(\ket{a})>E(\ket{b})$ \footnote{Note that a nontrivial catalytic transformation can never be implemented by local unitary transformations.}. In turn, if a supercatalytic transformation $\ket{a}\ket{c}\to\ket{b}\ket{d}$ and $\ket{d}\to\ket{c}$ is possible, then we must have that $E(\ket{d})>E(\ket{c})$ (since $c\neq d$) and $E(\ket{d})-E(\ket{c})\leq E(\ket{a})-E(\ket{b})$. Thus, the best Scrooge can hope for is to increase the entanglement entropy of his system by the same amount it decreases in the main system. Moreover, the choice of state that is borrowed from Scrooge is conditioned on the particular catalytic transformation that needs to be implemented. Thus, it is only natural that we measure how good Scrooge's strategy is relative to the choice of input and output states for which the loan is given. Hence, given any $\ket{a},\ket{b}\in S(\mathcal{H}_1)$ such that $\ket{a}\nrightarrow\ket{b}$ and given $\ket{c},\ket{d}\in S(\mathcal{H}_2)$ such that $\ket{a}\ket{c}\to\ket{b}\ket{d}$ and $\ket{d}\to\ket{c}$, we define the \emph{relative supercatalytic entanglement gain}, or \emph{gain} for short, as
\begin{equation}\label{gain}
G(\a,\b,\c,\d) = \frac{E(\d)-E(\c)}{E(\a)-E(\b)}.
\end{equation}
It follows from the discussion above that it always holds that $G(\a,\b,\c,\d)\in[0,1]$. Note that the gain being zero corresponds to the case of standard catalysis due to the strict Schur-concavity of entropy, while $G(\a,\b,\c,\d)>0$ corresponds to supercatalysis. The gain is one in and only in the best possible case, in which Scrooge obtains all the entanglement that is lost in the main system. We will thus measure the quality of Scrooge's loan strategies according to this figure of merit for transformations that are potentially supercatalyzable, which, as Proposition \ref{mainobs} shows, are those that are implementable catalytically. Therefore, we define the following quantities, where in an abuse of notation we denote these different notions of gain by $G_{\max}$ with the understanding that the precise case will be clear by the number of input arguments.

\begin{definition}\label{def:maximalgainc}
Given $\ket{a},\ket{b}\in S(\mathcal{H}_1)$ such that $\ket{a}\nrightarrow\ket{b}$ and $\ket{c}\in C(\ket{a},\ket{b})$, we define the \emph{maximal gain given input and output states $\ket{a}$ and $\ket{b}$ and lent state $\c$} by
\begin{equation}
G_{\max}(\a,\b,\c)=\sup\{G(\a,\b,\c,\d):\ket{a}\ket{c}\to\ket{b}\ket{d}, \ket{d}\to\ket{c}\}.
\end{equation}
\end{definition} 

\begin{definition}\label{def:maximalgain}
Given $\ket{a},\ket{b}\in S(\mathcal{H}_1)$ such that $\ket{a}\nrightarrow\ket{b}$ and $\ket{a} \rightarrow_c \ket{b}$, we define the \emph{maximal gain given input and output states $\ket{a}$ and $\ket{b}$} by
\begin{equation}
G_{\max}(\a,\b)=\sup\{G_{\max}(\a,\b,\c):\ket{c}\in C(\ket{a},\ket{b})\}.
\end{equation}
\end{definition} 

Thus, a given catalytic transformation $\ket{a} \rightarrow_c \ket{b}$ is supercatalyzable iff $G_{\max}(\a,\b)>0$, and this number gives the largest gain Scrooge can have over all states he can lend and all states he can ask for in return enabling the desired transformation in the main system. A strategy that achieves $G_{\max}(\a,\b)$ (if it exists) is hence optimal for Scrooge. On the other hand, $G_{\max}(\a,\b,\c)$ quantifies the largest gain by Scrooge on a given choice of lent state, and a given catalytic transformation $\ket{a} \rightarrow_c \ket{b}$ is fully supercatalyzable iff $G_{\max}(\a,\b,\c)>0$ holds $\forall \c\in C(\ket{a},\ket{b})$.

For reasons that will be clearer later, we are also going to study supercatalytic transformations in the scenario that we refer to as \emph{minimal}, where we constrain the amount of entanglement in the auxiliary system.

\begin{definition}\label{def:minimal}
Given $\ket{a},\ket{b}\in S(\mathcal{H}_1)$ such that $\ket{a}\nrightarrow\ket{b}$ and $\ket{a} \rightarrow_c \ket{b}$, we define the \emph{maximal gain given input and output states $\ket{a}$ and $\ket{b}$ in the minimal scenario} by
\begin{equation}
\tilde{G}_{\max}(\a,\b)=\sup\{G_{\max}(\a,\b,\c):\ket{c}\in C_r(\ket{a},\ket{b})\neq\emptyset, C_s(\ket{a},\ket{b})=\emptyset\;\forall s<r\}.
\end{equation}
\end{definition} 

This means that in the minimal scenario we only consider strategies in which Scrooge lends states of the minimal possible Schmidt rank. Hence, whenever a catalytic transformation can be upgraded to a supercatalytic transformation with a borrowed state $\c$ fulfilling this condition we will say that the transformation is \emph{minimally supercatalyzable} and we will speak about \emph{minimal supercatalytic transformations}. 

It is important to note that the Schmidt rank is obviously multiplicative, i.e.\ $SR(\a\c)=SR(\a)SR(\c)$. This property together with the fact that the Schmidt rank is an entanglement measure (i.e.\ it cannot increase under LOCC transformations), automatically constrains the Schmidt rank of the possible states that Scrooge can get in return if the borrowed state has been fixed or if it must belong to a set of a given Schmidt rank (such as in the minimal scenario). We state this as a proposition since we will use this fact repeatedly throughout the manuscript.
\begin{proposition}\label{SR(d)bound}
Let $\ket{a},\ket{b}\in S(\mathcal{H}_1)$ such that $\ket{a}\nrightarrow\ket{b}$. If there exists a supercatalytic transformation from $\ket{a}$ into $\ket{b}$ with borrowed state $\ket{c}\in S(\mathcal{H}_2)$ and returned state $\ket{d}\in S(\mathcal{H}_2)$ (i.e.\ $\ket{a}\ket{c}\to\ket{b}\ket{d}$), then it must hold that
\begin{equation}
SR(\d)\leq\left\lfloor\frac{SR(\a)SR(\c)}{SR(\b)}\right\rfloor.
\end{equation}
\end{proposition}

\section{Existence of supercatalyzable and fully supercatalyzable transformations}\label{sec3}

\subsection{Supercatalyzable transformations}\label{sec3a}

We begin with a simple observation that shows that actually all catalytic transformations are supercatalyzable at maximal gain.

\begin{proposition}\label{th:allsuper}
For every $\ket{a},\ket{b}\in S(\mathcal{H}_1)$ such that $\ket{a}\nrightarrow\ket{b}$ and $\ket{a} \rightarrow_c \ket{b}$ the corresponding transformation is always supercatalyzable. Furthermore, it always holds that $G_{\max}(\a,\b)=1$ and there always exists a choice of loan states for Scrooge that achieve this gain.
\end{proposition}

\begin{proof}
Suppose $\c\in S(\mathcal{H}_2)$ has the property that $\c\in C(\a,\b)$. We can consider an enlarged realization of the auxiliary system such that $\mathcal{H}_2=\mathcal{H}'_2\otimes\mathcal{H}''_2$ with $\mathcal{H}'_2= \mathcal{H}'_{A_2} \otimes \mathcal{H}'_{B_2}$ and $\mathcal{H}''_{2} = \mathcal{H}''_{A_2} \otimes \mathcal{H}''_{B_2}$. Then, it is clear that if we take $\c\in S(\mathcal{H}'_2)$, any choice of state $\ket{\psi}\in S(\mathcal{H}''_2)$ fulfils that $\c\ket{\psi}\in C(\a,\b)$, as $b\otimes c\succ a\otimes c$ implies that $b\otimes c\otimes\psi\succ a\otimes c\otimes\psi$ for any $\psi\in\Delta_n$ for any $n\in\mathbb{N}$. Hence, it follows in particular that $\c\b\in C(\a,\b)$. Moreover, it is clear that $\a\c\b\to\b\c\a$, since $f_k(a\otimes c\otimes b)=f_k(b\otimes c\otimes a)$ $\forall k$; that $\c\b$ and $\c\a$ have different ordered Schmidt coefficients, since $\ket{a}\nrightarrow\ket{b}$; and that $\c\a\to\c\b$, since $\c\in C(\a,\b)$. Thus, $\a$ can be supercatalytically transformed into $\b$. In fact, the protocol only requires to locally swap the registers of $\mathcal{H}_{1}$ and $\mathcal{H}''_{2}$ and it is a local unitary transformation. It is obvious then that $G(\a,\b,\c\b,\c\a)=1$.
\end{proof}

This proposition completely closes all the questions we have posed in the previous introductory sections, as long as we have the ability to choose the auxiliary state at will. Moreover, at first sight, it might seem to render the study of optimal strategies trivial. Nevertheless, note that the construction used in the proof of Proposition \ref{th:allsuper} is rather contrived as the state lent by Scrooge already contains a copy of the main system's target state embedded in it. Moreover, the borrowed state contains an artificially large amount of entanglement while the transformation can be implemented in the standard catalytic scenario with much less of it. In many instances in life, unconstrained resources lead to an unbounded profit; however, it is usually the case that one aims at maximizing the profit under a limited amount of resources at their disposal. This restriction is inherent to the scenario that we are considering here. Entangled states are difficult to prepare and entanglement is always regarded as a precious resource. In practice we might need to assume that Scrooge does not have perfect control on the state he can prepare and/or that the auxiliary system only operates under a limited amount of entanglement. This leads then to the study of supercatalysis and maximal gain under restricted classes of states that can be borrowed. The minimal scenario introduced at the end of the previous section (see\ Definition \ref{def:minimal}) arises naturally in this context. If catalytic transformations occur by giving the systems a physical realization of a certain dimensionality, i.e.\ that has access to a certain number of quantum levels, it seems quite reasonable that we are constrained to use the same physical implementation to achieve supercatalysis.

In Sec.\ \ref{sec4} we will study supercatalytic transformations with a constrained borrowed state and, in particular, in the minimal scenario. However, it  might be in order to point out here that, contrary to the unconstrained case as given in Proposition \ref{th:allsuper}, not all catalytic transformations are minimally supercatalyzable. This follows from Theorem 3 in \cite{BandRoy}. Using it, this reference provides a particular example of states $\a$ and $\b$ such that $SR(\a)=5$, $SR(\b)=4$ and $C_2(\a,\b)\neq\emptyset$ for which supercatalysis cannot occur if $SR(\c)=SR(\d)=2$. Furthermore, due to Proposition \ref{SR(d)bound}, in this case it cannot hold either that $SR(\d)>2$ when $SR(\c)=2$. From this we then conclude that in this case $G_{\max}(\a,\b,\c)=0 \;\; \forall\c\in C_2(\a,\b)$ and, consequently, $\tilde{G}_{\max}(\a,\b)=0$.

\subsection{Fully supercatalyzable transformations}\label{sec3b}

We move now to the study of fully supercatalyzable transformations. The example of \cite{BandRoy} with which we finished the previous section to argue that not all catalytic transformations are minimally supercatalyzable obviously implies in particular that not all catalytic transformations are fully supercatalyzable. We have studied this question in further generality and we have not found any instances of fully supercatalyzable transformations. In fact, our main contribution in this section are several no-go results on the impossibility of full supercatalysis for very general classes of catalytic transformations. While this does not exclude that some particular transformations might be fully supercatalyzable, this shows that in general Scrooge's choice of lent state plays a drastic role in the advantage he can obtain. This is not only because he does not obtain any at all for certain choices of catalysts, but also because, due to continuity, his gain remains close to zero if he lends states in their vicinity. We state precisely in the following the aforementioned results and we conclude this section with the corresponding proofs.

\begin{theorem}\label{th:nofull1}
Let $\ket{a},\ket{b}\in S(\mathcal{H}_1)$ such that $\ket{a}\nrightarrow\ket{b}$ and $\ket{a} \rightarrow_c \ket{b}$. If $C_r(\a,\b)\neq\emptyset$ with 
\begin{equation}
SR(\b)-r(SR(\a)-SR(\b))>0,
\end{equation}
then the transformation from $\a$ to $\b$ is not fully supercatalyzable.
\end{theorem}

\begin{corollary}
Let $\ket{a},\ket{b}\in S(\mathcal{H}_1)$ such that $\ket{a}\nrightarrow\ket{b}$ and $\ket{a} \rightarrow_c \ket{b}$. If $SR(\a)=SR(\b)$, then the transformation from $\a$ to $\b$ is not fully supercatalyzable.
\end{corollary}

\begin{theorem}\label{th:nofull2}
Let $\ket{a},\ket{b}\in S(\mathcal{H}_1)$ such that $\ket{a}\nrightarrow\ket{b}$ and $\ket{a} \rightarrow_c \ket{b}$. If $C_2(\a,\b)\neq\emptyset$, then the transformation from $\a$ to $\b$ is not fully supercatalyzable.
\end{theorem}

\subsubsection{Proof of Theorem \ref{th:nofull1}}

The idea of this proof is to first show that the set of all catalysts with bounded Schmidt rank for any given catalytic transformation is compact (which is also relevant in later discussions). This in particular implies that there always exists a notion of most entangled catalyst in $C_r(\a,\b)$ for any $r$ whenever this set is not empty. This then entails that supercatalysis is impossible if we fix this state to be the borrowed state, i.e.\ $\c$, if we furthermore restrict the returned state, i.e.\ $\d$, to fulfil $\d\in C_r(\a,\b)$ as well. Finally, if the Schmidt ranks of $\a$, $\b$ and $\c$ are chosen properly so that, by Proposition \ref{SR(d)bound}, it is impossible that the returned state fulfils $\d\notin C_r(\a,\b)$, we then obtain that the catalytic transformation is not supercatalyzable whenever the state $\c$ is lent.

\begin{lemma}
\label{lemma:compact}
Let $\ket{a},\ket{b}\in S(\mathcal{H}_1)$ such that $\ket{a} \rightarrow_c \ket{b}$. Then, for any $r\in\mathbb{N}$ it holds that $C_r(\a,\b)$ is compact. 
\end{lemma}
\begin{proof}
Given $\ket{a},\ket{b}\in S(\mathcal{H}_1)$ with $\dim\mathcal{H}_{A_1}=\dim\mathcal{H}_{B_1}=d$ and $r\in\mathbb{N}$, we can regard $C_r(\a,\b)\subset S(\mathcal{H}_2)$ with $\dim\mathcal{H}_{A_2}=\dim\mathcal{H}_{B_2}=r$. Therefore, to prove the claim it is enough to show that this set is closed and bounded. The latter property obviously holds, so we only have to verify the first one. In order to do so, consider the manifestly continuous map $s:S(\mathcal{H}_2)\to\Delta'_r$ that assigns to every state the vector of its ordered Schmidt coefficients. Since the preimage of a closed set under a continuous map must be closed, it is therefore sufficient to prove that $s(C_r(\a,\b))$ is closed, i.e.\ that $\Delta'_r \backslash s(C_r(\a,\b))$ is open in $\Delta'_r$. In the following we can assume that $s(C_r(\a,\b))\neq\emptyset$ and $s(C_r(\a,\b))\neq\Delta'_r$ as the claim is otherwise trivially true. Hence, take any $c\in \Delta'_r \backslash s(C_r(\a,\b))$, by Nielsen's theorem there must exist some value of $k\in\{1,2,\ldots,dr-1\}$ such that  
\begin{equation}
f_k (a \otimes c) > f_k (b \otimes c).
\end{equation}
Since $f_k$ is a continuous function, it follows that $\exists\epsilon>0$ such that $\forall c_\epsilon\in\Delta'_r$ for which $||c-c_\epsilon||<\epsilon$, it then holds that $f_k(a\otimes c_\epsilon)>f_k(b\otimes c_\epsilon)$, and, consequently, that $c_\epsilon\in \Delta'_r \backslash s(C_r(\a,\b))$.
\end{proof}

\begin{definition}
Given $\ket{a},\ket{b}\in S(\mathcal{H}_1)$ such that $\ket{a} \rightarrow_c \ket{b}$ and $r\in\mathbb{N}$ such that $C_r(\a,\b)\neq\emptyset$, we define
\begin{equation}
E_r(\a,\b)=\max\{E(\c):\c\in C_r(\a,\b)\}.
\end{equation}
\end{definition}

Note that we write $\max$ instead of $\sup$ because we are optimizing a continuous function -- the entanglement entropy -- over a compact set, $C_r(\a,\b)$, as per Lemma \ref{lemma:compact}. Thus, for any pair $\ket{a},\ket{b}\in S(\mathcal{H}_1)$ and $r\in\mathbb{N}$ such that $C_r(\a,\b)\neq\emptyset$, this guarantees that there exists at least one choice of $\c\in C_r(\a,\b)$ such that $E_r(\a,\b)=E(\c)$. We refer to any state with this property as \emph{a most entangled catalyst of Schmidt rank less than or equal to $r$ for the transformation} $\ket{a} \rightarrow_c \ket{b}$. For $r=2$, given that majorization induces a total order in this case, this state is unique. Moreover, thanks to Theorem \ref{th:catalyst}, if in addition it holds that $SR(|a\rangle),SR(|b\rangle)\leq4$, then the vector of Schmidt coefficients of the most entangled catalyst of Schmidt rank equal to 2 is given by $(x_{\min}(a,b),1-x_{\min}(a,b))$ in the notation used therein. By the same reasons, we can also speak about \emph{a least entangled catalyst of Schmidt rank less than or equal to $r$ for the transformation} $\ket{a} \rightarrow_c \ket{b}$. Analogously, when the premises of Theorem \ref{th:catalyst} hold, the Schmidt vector of the least entangled catalyst of Schmidt rank equal to 2 is given by $(x_{\max}(a,b),1-x_{\max}(a,b))$.

\begin{remark}\label{remark}
Note that, due to Proposition \ref{th:allsuper}, we can also write $\max$ instead of $\sup$ in the definition of the maximal gain given in Definition \ref{def:maximalgain}. Now we can see that the same holds for the other notions of maximal gain given in that section (Definition \ref{def:maximalgainc} and Definition \ref{def:minimal}). Note first that $G(\a,\b,\c,\d)$ is continuous in its domain due to the continuity of the entanglement entropy. By Proposition \ref{SR(d)bound}, in order to compute $G_{\max}(\a,\b,\c)$ given $\a$, $\b$ and $\c$ we can restrict the optimization to states $\d$ such that $SR(\d)\leq\lfloor SR(\a)SR(\c)/SR(\b)\rfloor$. Now, by the same reasoning used in the proof of Lemma \ref{lemma:compact}, one can readily see that, given $r\in\mathbb{N}$, the sets $\{\d \in S(\mathcal{H}_2): \ket{a}\ket{c}\to\ket{b}\ket{d}, SR(\d)\leq r\}$ and $\{\d\in S(\mathcal{H}_2): \ket{d}\to\ket{c}, SR(\d)\leq r\}$ are closed, and, hence, so is their intersection, which is moreover clearly bounded. Thus, in Definition \ref{def:maximalgainc} we are also optimizing a continuous function over a compact set. The case of the maximal gain in the minimal scenario (see\ Definition \ref{def:minimal}) follows immediately from Lemma \ref{lemma:compact} given that the Schmidt rank of the borrowed state $\c$ is bounded. 
\end{remark}

\begin{proof}[Proof of Theorem~\ref{th:nofull1}]
Note that the condition $SR(\b)-r(SR(\a)-SR(\b))>0$ implies
\begin{equation}
\frac{SR(\a)r}{SR(\b)}<r+1.
\end{equation}
By assumption, $C_r(\a,\b)\neq\emptyset$, and, then, if we lend any catalyst in this set to supercatalytically transform $\a$ into $\b$, the condition above together with Proposition \ref{SR(d)bound} mean that
\begin{equation}
SR(\d)\leq\left\lfloor\frac{SR(\a)SR(\c)}{SR(\b)}\right\rfloor \leq \left\lfloor\frac{SR(\a)r}{SR(\b)}\right\rfloor<r+1,
\end{equation}
i.e. the returned state $\d$ must satisfy $SR(\d)\leq r$. Hence, by Proposition \ref{mainobs}, $\d\in C_r(\a,\b)$. Let us then assume for a contradiction that a transformation fulfilling the premises of the theorem is fully supercatalyzable. Then, in particular, supercatalysis is possible choosing the lent state $\c$ to be a most entangled catalyst of Schmidt rank less than or equal to $r$ for the transformation $\ket{a} \rightarrow_c \ket{b}$. This then means that there exists $\d\in C_r(\a,\b)$ such that $\ket{a}\ket{c}\to\ket{b}\ket{d}$ and $\ket{d}\to\ket{c}$ with $c\neq d$. By the strict Schur-concavity of the entropy the latter condition implies that $E(\d)>E(\c)=E_r(\a,\b)$, which is a contradiction.
\end{proof}

\subsubsection{Proof of Theorem \ref{th:nofull2}}

The proof of this theorem uses similar ideas to what we have used in the proof of Theorem \ref{th:nofull1}. Due to Lemma \ref{lemma:compact}, if $C_2(\a,\b)\neq\emptyset$, then full supercatalysis is only possible if when we borrow the most entangled state of Schmidt rank equal to 2 we can perform supercatalysis returning a state of higher Schmidt rank. We first show that if this is the case, then without loss of generality such supercatalytic transformation must be possible returning a state with Schmidt rank equal to 3. We then show that this is nevertheless impossible.

\begin{lemma} \label{lemma1:nofull2}
Let $\ket{c},\ket{d}\in S(\mathcal{H}_2)$ such that $SR(\d)>3$, $SR(\c)=2$, and $\d \rightarrow \c$. Then, there exists $\ket{d'}\in S(\mathcal{H}_2)$ with $SR(\ket{d'})=3$ such that $\d \rightarrow \ket{d'}$ and $\ket{d'} \rightarrow \c$.
\end{lemma}
\begin{proof}
Let $d=(d_1,d_2,d_3,d_4,...)$, $c=(c_1,c_2)$ and $d'=(d'_1,d'_2,d'_3)$. By Nielsen's theorem, there are three necessary and sufficient conditions for such $\ket{d'}$ to exist. On the one hand, $\ket{d'} \rightarrow \c$ iff $d'_1 \leq c_1$ (i). On the other hand,  $\d \rightarrow \ket{d'}$ iff $d_1 \leq d'_1$ (ii) and $d_1 + d_2 \leq d'_1 + d'_2$ (iii). To prove the claim, we make in the following a choice of $d'$ for any given $d$ and $c$ that is compatible with these three conditions. First, we fix $d'_1=c_1$, which obviously satisfies condition (i) and also condition (ii) (since the fact that $\d \rightarrow \c$ translates to $d_1 \leq c_1$). Note additionally that, since $c_1 \geq 1/2$, this choice is consistent with the fact that $d'_1 \geq d'_2$ and $d'_1 \geq d'_3$ should hold. Thus, it only remains to make a good choice of $d'_2$ so that condition (iii) holds as well. For this, we write $d' = (c_1, \frac{1-c_1}{2}+\alpha, \frac{1-c_1}{2}-\alpha)$, provided that the following requirements are met:
\begin{equation} 
\alpha \leq \frac{3}{2} c_1 -\frac{1}{2}
\end{equation}
so that $d'_1 \geq d'_2$, $\alpha \geq 0$ so that $d'_2 \geq d'_3$, and
\begin{equation}\label{eq1:lemma1:nofull2}
\alpha < \frac{1-c_1}{2}
\end{equation}
so that $d'_3 > 0$. It turns out that the most restrictive upper bound on $\alpha$ is the latter, because $\frac{1-c_1}{2} \leq \frac{3}{2} c_1 -\frac{1}{2}$ is equivalent to $c_1\geq1/2$, which is always true. On the other hand, with this parametrization condition (iii) boils down to
\begin{equation}\label{eq2:lemma1:nofull2}
\alpha \geq d_1+d_2 - \frac{c_1}{2}-\frac{1}{2}.
\end{equation}
Thus, to finish the proof we only need to provide a choice of $\alpha \geq 0$ such that both Eqs.\ (\ref{eq1:lemma1:nofull2}) and (\ref{eq2:lemma1:nofull2}) hold. Such a choice is given by $\alpha = \max \{0, d_1+d_2 - c_1/2-1/2\}$ (note that the fact that $d_1+d_2<1$ guarantees that the condition of Eq.\ (\ref{eq1:lemma1:nofull2}) is met).
\end{proof}

Note that if a given catalytic transformation $\ket{a} \rightarrow_c \ket{b}$ such that $C_2(\a,\b)\neq\emptyset$ is supercatalyzable with a borrowed state $\c\in C_2(\a,\b)$ and a returned state $\d$ such that $SR(\d)>3$, i.e.\ $\a\c\to\b\d$ and $\d\to\c$, then Lemma \ref{lemma1:nofull2} entails that $\a\c\to\b\d\to\b\ket{d'}$ and $\ket{d'} \rightarrow \c$ with $SR(\ket{d'})=3$. That is, the transformation has to be supercatalyzable as well with the same borrowed state and a Schmidt-rank-3 returned state. The following lemmas show that this is impossible when the borrowed state is the most entangled state of Schmidt rank equal to 2. Before addressing this, let us fix first some further notation. Note that given $\ket{a}\in S(\mathcal{H}_1)$ and $\ket{c}\in S(\mathcal{H}_2)$ such that $SR(\c)=2$, then for every given $k\in\mathbb{N}$ there exist $k_1\in\mathbb{N}$ and $k_2\in\mathbb{N}_0$ fulfilling $k_1\geq k_2$ and $k_1+k_2=k$ such that
\begin{equation}
f_k(a\otimes c)=f_{k_1,k_2}(a\otimes c):=c_1\sum_{i=1}^{k_1}a_i+c_2\sum_{i=1}^{k_2}a_i,
\end{equation}
where it should be understood that $\sum_{i=1}^0a_i=0$ when $k_2=0$. It is important to note that for a given $k$ the choice of the pair $\{k_1,k_2\}$ needs not be unique (this is the case when certain entries in $a\otimes c$ happen to be equal). Note as well that it must hold that 
\begin{equation}\label{eq:4}
f_k(a\otimes c)\geq f_{m,n}(a\otimes c)
\end{equation} 
for all $m\in\mathbb{N}$ and $n\in\mathbb{N}_0$ such that $m+n=k$.

\begin{lemma}\label{lemma2:nofull2}
Given $\ket{a},\ket{b}\in S(\mathcal{H}_1)$ such that $\ket{a}\nrightarrow\ket{b}$ and $C_2(\a,\b)\neq\emptyset$, let $\c$ be its most entangled catalyst of Schmidt rank equal to 2. Then, if $\a\c\to \b\d$, $\d\to \c$ and $SR(\d)=3$, it must hold that $d_1=c_1$.
\end{lemma}
\begin{proof}
On the one hand, $\d\to\c$ is equivalent to $d_1\leq c_1$. On the other hand, by Nielsen's theorem, it always holds that $\d\to\ket{c'}$ where $\ket{c'}$ is the Schmidt-rank-2 state with Schmidt vector $(d_1,1-d_1)$. This implies that $\a\c\to \b\d\to\b\ket{c'}$ and $\ket{c'}\to \c$. Thus, if $d_1<c_1$, supercatalysis from $\a$ to $\b$ would be possible with $\c$ as a borrowed state and $\ket{c'}$ as a returned state with $\ket{c'}\in C_2(\a,\b)$ and $E(\ket{c'})>E(\c)$. Yet, this is a contradiction with the assumption that $\c$ is the most entangled catalyst of Schmidt rank equal to 2 for the catalytic transformation $\ket{a} \rightarrow_c \ket{b}$.
\end{proof}

\begin{lemma}\label{lemma3:nofull2}
Given $\ket{a},\ket{b}\in S(\mathcal{H}_1)$ such that $\ket{a}\nrightarrow\ket{b}$ and $C_2(\a,\b)\neq\emptyset$, let $\c$ be its most entangled catalyst of Schmidt rank equal to 2. Then, if $\a\c\to \b\d$, $\d\to \c$ and $SR(\d)=3$, it must hold that $\a\c\to \b\ket{c_\epsilon^{(3)}}$ and $\ket{c_\epsilon^{(3)}}\to \c$ for all $\epsilon>0$ sufficiently small, where $c_\epsilon^{(3)}:=(c_1,c_2-\epsilon,\epsilon)$.
\end{lemma}
\begin{proof}
Note that $c_2>0$ since a separable state cannot be catalyst and, therefore, $c_\epsilon^{(3)}$ is well defined and its entries are ordered for $\epsilon$ small enough. The claim follows immediately from Lemma \ref{lemma2:nofull2} as, by Nielsen's theorem, it implies that $\d\to \ket{c_\epsilon^{(3)}}$ for all $\epsilon>0$ sufficiently small; while $\ket{c_\epsilon^{(3)}}\to \c$ holds trivially.
\end{proof}

\begin{lemma}\label{lemma4:nofull2}
Given $\ket{a},\ket{b}\in S(\mathcal{H}_1)$ such that $\ket{a}\nrightarrow\ket{b}$ and $C_2(\a,\b)\neq\emptyset$, let $\c$ be its most entangled catalyst of Schmidt rank equal to 2. Then, it cannot hold that $\a\c\to \b\d$, $\d\to \c$ and $SR(\d)=3$.
\end{lemma}
\begin{proof}
Given $\c\in S(\mathcal{H}_2)$ to be the most entangled catalyst of Schmidt rank equal to 2 for the catalytic transformation $\ket{a} \rightarrow_c \ket{b}$, for $\epsilon>0$ sufficiently small let $\ket{c_\epsilon^{(2)}}$ be a state in $S(\mathcal{H}_2)$ with Schmidt vector given by $c_\epsilon^{(2)}:=(c_1-\epsilon,c_2+\epsilon)$ (note that, since the maximally entangled state cannot be a catalyst, it holds that $c_1>1/2$ and, hence, the entries of $c_\epsilon^{(2)}$ remain ordered for small enough $\epsilon$). Since for all $\epsilon>0$ sufficiently small it holds that $E\left(\ket{c_\epsilon^{(2)}}\right)>E(\c)$, we have that this state is not a catalyst for $a\to_c b$ under the given assumption on $\epsilon$. Thus, since $\a\ket{c_\epsilon^{(2)}}\nrightarrow \b\ket{c_\epsilon^{(2)}}$, there must exist a fixed $k\in\mathbb{N}$ such that $f_k(a\otimes c_\epsilon^{(2)})>f_k(b\otimes c_\epsilon^{(2)})$ for all $\epsilon>0$ sufficiently small. That is, there exist $k_1,k_2,k'_1,k'_2$ as explained above such that
\begin{equation}\label{assumption}
f_k(a\otimes c_\epsilon^{(2)})=f_{k_1,k_2}(a\otimes c_\epsilon^{(2)})=(c_1-\epsilon)\sum_{i=1}^{k_1}a_i+(c_2+\epsilon)\sum_{i=1}^{k_2}a_i>f_k(b\otimes c_\epsilon^{(2)})=f_{k'_1,k'_2}(b\otimes c_\epsilon^{(2)})=(c_1-\epsilon)\sum_{i=1}^{k'_1}b_i+(c_2+\epsilon)\sum_{i=1}^{k'_2}b_i.
\end{equation}
By continuity and the fact that it must hold that $f_k(a\otimes c)\leq f_k(b\otimes c)$ given that $\c\in C(\a,\b)$, it follows that 
\begin{equation}\label{eq:equality}
f_k(a\otimes c)=f_{k_1,k_2}(a\otimes c)=f_k(b\otimes c)= f_{k'_1,k'_2}(b\otimes c).
\end{equation}

Assume now for a contradiction that there exists $\d$ such that $\a\c\to \b\d$, $\d\to \c$ and $SR(\d)=3$. Then, by Lemma \ref{lemma3:nofull2}, $f_{k}(a\otimes c)\leq f_{k}(b\otimes c_\epsilon^{(3)})\leq f_{k}(b\otimes c)$ holds for $\epsilon$ small enough, where the latter inequality is due to $\ket{c_\epsilon^{(3)}}\to \c$. Hence, by Eq.\ (\ref{eq:equality}), we have that 
\begin{equation}\label{eq:nueva}
f_k(a\otimes c)= f_{k}(b\otimes c_\epsilon^{(3)})=f_k(b\otimes c).
\end{equation}
Now, by taking $\epsilon$ sufficiently small we can guarantee that $(c_2-\epsilon)b_{SR(|b\rangle)} > \epsilon b_1$. Therefore, there exist $k''_1\in\mathbb{N}$ and $k''_2\in\mathbb{N}_0$ with $k''_1+k''_2=k$ such that for all $\epsilon>0$ small enough it holds that
\begin{equation}\label{eq:sandwhich}
f_{k}(b\otimes c_\epsilon^{(3)})=c_1\sum_{i=1}^{k''_1}b_i+(c_2-\epsilon)\sum_{i=1}^{k''_2}b_i=f_{k_1'',k_2''}(b\otimes c)-\epsilon\sum_{i=1}^{k''_2}b_i.
\end{equation}
However, $f_{k_1'', k_2''}(b\otimes c)\leq f_k(b\otimes c)$ by Eq.\ (\ref{eq:4}), and Eq.\ (\ref{eq:nueva}) is then impossible unless $k_2''=0$. Therefore, 
\begin{equation}
f_k(b\otimes c)= f_{k,0}(b\otimes c)=c_1\sum_{i=1}^{k}b_i,
\end{equation}
which in particular entails that $c_1b_k\geq c_2b_1$ and that 
\begin{equation}\label{eq:1}
f_{k'}(b\otimes c)= f_{k',0}(b\otimes c)\quad\forall k'\leq k.
\end{equation}
We now establish the following chain of inequalities:
\begin{equation}\label{ineqs}
c_2\sum_{i=1}^{k_2}a_i\leq c_2\sum_{i=1}^{k_2}a_1\leq c_2\sum_{i=1}^{k_2}b_1\leq c_1\sum_{i=1}^{k_2}b_k\leq c_1\sum_{i=1}^{k_2}b_{i+k_1},
\end{equation}
where in the first inequality we have used that $a_1\geq a_i$ $\forall i$, in the second that $a_1\leq b_1$ is a necessary condition for $\a\to_c \b$, in the third the aforementioned observation that $c_1b_k\geq c_2b_1$ and in the fourth that $k_1+k_2=k$ and $b_k\leq b_i$ for $k_1+1\leq i\leq k$. But
\begin{equation}\label{eq:sums}
f_{k_1,k_2}(a\otimes c)=f_{k,0}(b\otimes c)\Leftrightarrow c_1\sum_{i=1}^{k_1}a_i+c_2\sum_{i=1}^{k_2}a_i=c_1\sum_{i=1}^{k_1}b_i+c_1\sum_{i=1}^{k_2}b_{i+k_1},
\end{equation}
which, together with Eq.\ (\ref{ineqs}), requires that
\begin{equation}\label{eq:above2}
\sum_{i=1}^{k_1}a_i\geq\sum_{i=1}^{k_1}b_i.
\end{equation}
However, using Eqs.\ (\ref{eq:4}) and (\ref{eq:1}), this implies that
\begin{equation}\label{eq:above}
f_{k_1}(a\otimes c)\geq f_{k_1,0}(a\otimes c)=c_1\sum_{i=1}^{k_1}a_i\geq c_1\sum_{i=1}^{k_1}b_i=f_{k_1}(b\otimes c),
\end{equation}
while $f_{k_1}(a\otimes c)\leq f_{k_1}(b\otimes c)$ must hold since $\a\c\to \b\c$. Thus, we conclude that all the inequalities in Eqs.\ (\ref{ineqs}), (\ref{eq:above2}) and (\ref{eq:above}) must hold with equality. In particular, from Eq.\ (\ref{eq:above2}) we obtain that
\begin{equation}\label{eq:2}
\sum_{i=1}^{k_1}a_i=\sum_{i=1}^{k_1}b_i,
\end{equation}
and from Eq.\ (\ref{ineqs}) that $b_1=a_1=a_i$ for $1\leq i\leq k_2$. This last equation entails that for $1\leq j\leq k_2$ it holds that $f_j(a\otimes c)=ja_1c_1=jb_1c_1$ and, therefore, $f_j(a\otimes c)\leq f_j(b\otimes c)$ can only hold if $b_1=a_1=a_i=b_i$ for $1\leq i\leq k_2$ (recall that $c_1>c_2$). This in particular means that
\begin{equation}\label{eq:3}
\sum_{i=1}^{k_2}a_i=\sum_{i=1}^{k_2}b_i.
\end{equation}
Now, using Eqs.\ (\ref{eq:4}), (\ref{eq:2}) and (\ref{eq:3}) we arrive at the conclusion that
\begin{equation}
f_k(b\otimes c_\epsilon^{(2)})=f_{k'_1,k'_2}(b\otimes c_\epsilon^{(2)})\geq f_{k_1,k_2}(b\otimes c_\epsilon^{(2)})=f_{k_1,k_2}(a\otimes c_\epsilon^{(2)})=f_{k}(a\otimes c_\epsilon^{(2)})
\end{equation}
holds for all $\epsilon>0$ sufficiently small, which is a contradiction with Eq.\ (\ref{assumption}). This concludes the proof.
\end{proof}

\begin{proof}[Proof of Theorem~\ref{th:nofull2}]
Assume for a contradiction that there exist $\ket{a},\ket{b}\in S(\mathcal{H}_1)$ such that $\ket{a}\nrightarrow\ket{b}$ and $C_2(\a,\b)\neq\emptyset$, for which the transformation from $\a$ to $\b$ is fully supercatalyzable. Then, in particular, a supercatalytic transformation from $\a$ to $\b$ must be possible when the borrowed state $\c$ is the most entangled catalyst of Schmidt rank equal to $2$. However, by Lemmas \ref{lemma1:nofull2} and \ref{lemma4:nofull2} this imposes that the corresponding returned state $\d$ must have the property that $\d\in C_2(\a,\b)$ as well. Since it must hold that $\ket{d}\to\ket{c}$ with $c\neq d$, the strict Schur-concavity of the entropy enforces that $E(\d)>E(\c)=E_2(\a,\b)$, which is a contradiction.
\end{proof}

\section{Supercatalytic protocols with constrained auxiliary states: minimal scenario}\label{sec4}

As we have discussed in Sec.\ \ref{sec3a}, the problem of optimizing the gain trivializes when Scrooge has absolute freedom to choose the state he lends. We have also argued therein that it is much more natural to consider that this choice is constrained, such as in the minimal scenario. Thus, given a fixed non-trivial catalytic transformation $\ket{a} \rightarrow_c \ket{b}$, in this section we want to explore what $G_{\max}(\a,\b,\c)$ could be for specific choices of borrowed states $\c$. In particular, we would like to optimize this quantity under the condition that the borrowed states belong to a certain set; a natural choice being the optimal gain in the minimal scenario $\tilde{G}_{\max}(\a,\b)$, where this set is the set of all catalysts of the minimal possible Schmidt rank.

The analysis performed in Sec.\ \ref{sec3b} in the study of full supercatalysis shows that, in the scenarios where the Schmidt rank of the borrowed state is constrained, what could be called a generous strategy, in which Scrooge lends as much entanglement as possible, is in general a bad strategy. This is because we have seen that under quite general premises lending a most entangled catalyst of at most a given Schmidt rank leads to zero gain and, by continuity, states in its vicinity must have gain close to zero. We can put this statement in quantitative terms with the following relatively immediate upper bound on $G_{\max}(\a,\b,\c)$.

\begin{theorem}\label{th:bound}
If there exists a supercatalytic transformation from $\ket{a}\in S(\mathcal{H}_1)$ into $\ket{b}\in S(\mathcal{H}_1)$ with borrowed state $\ket{c}\in S(\mathcal{H}_2)$, it then holds that
\begin{equation}
G_{max}(\a,\b,\c) \leq \frac{E_{\left\lfloor SR(\a)SR(\c)/SR(\b)\right\rfloor}(\a,\b)-E(\c)}{E(\a)-E(\b)}.
\label{bound}
\end{equation}
\end{theorem}
\begin{proof}
From Propositions \ref{mainobs} and \ref{SR(d)bound}, it follows that the returned state $\d$ must satisfy 
$$\d\in C_{\left\lfloor SR(\a)SR(\c)/SR(\b)\right\rfloor}(\a,\b).$$ Thus, by definition, we have that $E(\d)\leq E_{\left\lfloor SR(\a)SR(\c)/SR(\b)\right\rfloor}(\a,\b)$.
\end{proof}

Roughly speaking, the aforementioned observation that generous strategies perform in general badly leaves Scrooge with two possibilities. Is it in general preferable for him to follow a miserly strategy in which he lends as little entanglement as possible? Or, on the contrary, is there a trade-off that favours an intermediate strategy in which he lends more entanglement -- but not as much as possible -- in order to increase his gain? Note that if the bound of Theorem \ref{th:bound} were tight for all borrowed states $\c$ of at most a given Schmidt rank, then this would mean that, under this constraint, Scrooge's optimal strategy is to lend a least entangled catalyst in this set.

\subsection{Exploring optimal protocols in the minimal scenario}

If we fix the Schmidt rank of the states Scrooge can lend, then by virtue of Propositions \ref{mainobs} and \ref{SR(d)bound} we know that all possible strategies rely on optimizations over borrowed states in $C_r(\a,\b)$ and returned states in $C_{r'}(\a,\b)$ for some $r,r'\in\mathbb{N}$. However, as mentioned in the introduction, the general characterization of such sets seems a daunting problem and the available results in this direction are limited. For this reason, in this section we consider the simplest scenario in which $SR(\a),SR(\b)\leq4$, where Theorem \ref{th:catalyst} characterizes $C_2(\a,\b)$. In fact, in this case, whenever $C_2(\a,\b)\neq\emptyset$, due to Proposition \ref{SR(d)bound}, in the minimal scenario it must also hold that $\d\in C_2(\a,\b)$. Thus, by the results of Sec.\ \ref{sec3b}, we know that in this case a generous strategy is bound to failure. Moreover, in the language of Theorem \ref{th:catalyst}, given $\ket{a},\ket{b}\in S(\mathcal{H}_1)$ such that $SR(|a\rangle),SR(|b\rangle)\leq4$ and $\ket{a}\nrightarrow\ket{b}$ fulfilling $x_{\min}(a,b)\leq x_{\max}(a,b)$, the bound of Theorem \ref{th:bound} reduces to
\begin{equation}\label{minbound}
G_{max}(\a,\b,\c) \leq \frac{h(x_{\min}(\a,\b))-E(\c)}{E(\a)-E(\b)},
\end{equation}
where $h(\cdot)$ stands for the binary entropy, and we have that
\begin{equation}\label{optminbound}
\tilde{G}_{max}(\a,\b) \leq \frac{h(x_{\min}(\a,\b))-h(x_{\max}(\a,\b))}{E(\a)-E(\b)}.
\end{equation}
Note that this latter bound is optimal iff the miserly strategy is optimal and a supercatalytic transformation from $\a$ to $\b$ is possible borrowing the least entangled catalyst of Schmidt rank equal to 2 and returning the most entangled catalyst of this Schmidt rank. Furthermore, if this is the case, then the bound of Eq.\ (\ref{minbound}) is optimal $\forall\c\in C_2(\a,\b)$. This is because the LOCC ordering in this case is total and, hence, given any such $\c$ we can always transform it into the least entangled catalyst. Thus, we can implement this transformation as a first step of the supercatalytic protocol and then, as a second step, run the supercatalytic protocol that returns the most entangled catalyst upon borrowing the least entangled catalyst. Actually, whenever the bound in Eq.\ (\ref{minbound}) is optimal for some choice of $\c\in C_2(\a,\b)$, then, by the same reason as above, it has to remain optimal for every $\ket{c'}\in C_2(\a,\b)$ such that $E(\ket{c'})>E(\c)$.

We have performed an extensive investigation assisted by numerical means of what $G_{\max}(\a,\b,\c)$ and $\tilde{G}_{\max}(\a,\b)$ can be in this simplest scenario and we have found that in most instances the miserly strategy performs well. However, it can sometimes be not only suboptimal but even extremely disadvantageous. This lack of structure, even in the simplest conceivable case, makes us think that determining in general the optimal supercatalytic strategies and gain in the minimal scenario for arbitrary input and output states in the main system appears to be an intractable problem. In the following we provide a few selected examples to illustrate these points. The figures plot $G_{\max}(\a,\b,\c)$ (together with the bound given by Eq.\ (\ref{minbound})) vs.\ the entanglement entropy of the borrowed state, $E(\c)$, in the whole range of possible catalysts, $[h(x_{\max}(a,b),h(x_{\min}(a,b))]$, for different choices of catalytic transformation $\ket{a} \rightarrow_c \ket{b}$.    

\begin{example}
Figure \ref{fig:Ej_JP} corresponds to Jonathan and Plenio's original example of a catalytic transformation \cite{JP}: $a = (0.4, 0.4, 0.1, 0.1)$ and $b=(0.5, 0.25, 0.25, 0)$. This example shows that the bound of Eq.\ (\ref{optminbound}) can be optimal, and, consequently, so is the miserly strategy. Note that in order to see this, it is enough to verify that 
\begin{equation}
b\otimes(x_{\min}(a,b),1-x_{\min}(a,b))\succ a\otimes(x_{\max}(a,b),1-x_{\max}(a,b)),
\end{equation}
which is straightforward to check.
\end{example}

\begin{figure}[h!]
    \centering
    \begin{minipage}{0.48\textwidth} 
        \centering
        \includegraphics[width=\textwidth]{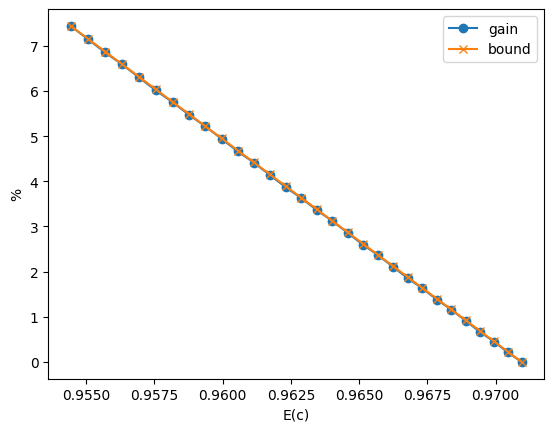}
        \caption{\label{fig:Ej_JP} $G_{\max}(\a,\b,\c)$ (together with its upper bound given in Eq.\ (\ref{minbound})) vs.\ the entanglement entropy of the borrowed state, $E(\c)$, for Example 1.}
    \end{minipage}
    \hfill 
    \begin{minipage}{0.48\textwidth}
        \centering
        \includegraphics[width=\textwidth]{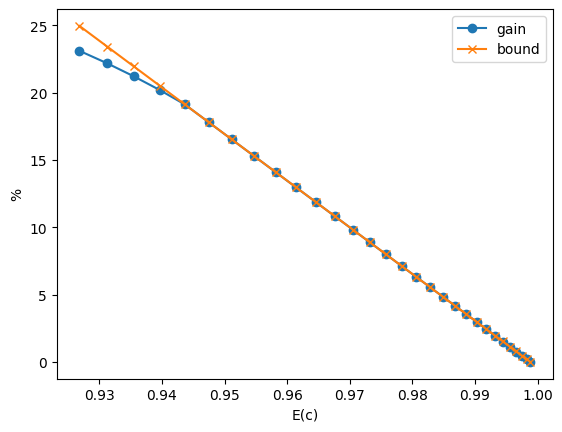}
        \caption{\label{fig:Ej_FWX} $G_{\max}(\a,\b,\c)$ (together with its upper bound given in Eq.\ (\ref{minbound})) vs.\ the entanglement entropy of the borrowed state, $E(\c)$, for Example 2.}
    \end{minipage}
\end{figure}

\begin{example} In this example (which was also considered in \cite{FWX}), we have perturbed the Schmidt vector of the input state of Example 1: $a = (0.4, 0.36, 0.14, 0.1)$ and $b=(0.5, 0.25, 0.25, 0)$. The result is depicted in Fig.\ \ref{fig:Ej_FWX} and it shows that, although the bound of Eq.\ (\ref{optminbound}) is not tight in this case (and, hence, nor is that of Eq.\ (\ref{minbound}) for the whole range of catalysts), the miserly strategy remains optimal.
\end{example}

\begin{example} In this example we further perturb the Schmidt vector of the input state of the previous examples: $a = (0.41, 0.38, 0.12, 0.09)$ and $b=(0.5, 0.25, 0.25, 0)$. This leads to a drop in the gain of the miserly strategy to the point that it is (slightly) suboptimal, as can be seen in Fig.\ \ref{fig:Ej_caida}.
\end{example}

\begin{example}
Here, by choosing $a = (0.88, 0.08, 0.02, 0.02)$ and $b = (0.9, 0.05, 0.05, 0)$, we construct an example where the miserly strategy is not only suboptimal but leads to a complete failure, as seen in Fig.\ \ref{fig:Ej_batacazo}. This shows that lending the least entangled catalyst can lead to zero gain, just as it always does in the cases we study here when lending the most entangled catalyst. On the other hand, an intermediate strategy makes it possible to obtain a non-negligible gain.
\end{example}

\begin{figure}[h!]
    \centering
    \begin{minipage}{0.48\textwidth} 
        \centering
        \includegraphics[width=\textwidth]{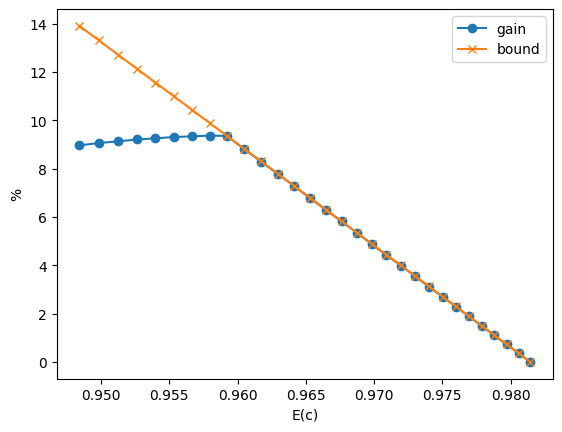}
        \caption{\label{fig:Ej_caida} $G_{\max}(\a,\b,\c)$ (together with its upper bound given in Eq.\ (\ref{minbound})) vs.\ the entanglement entropy of the borrowed state, $E(\c)$, for Example 3.}
    \end{minipage}
    \hfill 
    \begin{minipage}{0.48\textwidth}
        \centering
        \includegraphics[width=\textwidth]{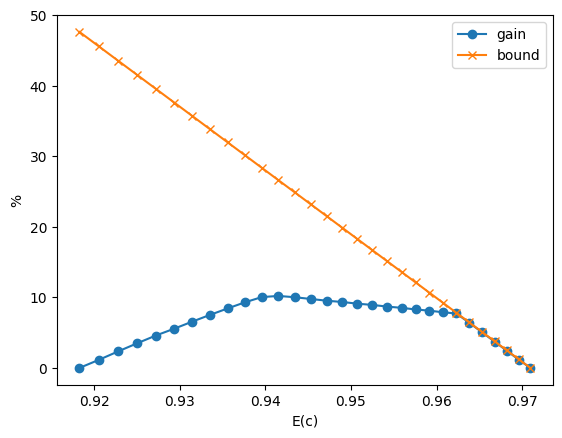}
        \caption{\label{fig:Ej_batacazo} $G_{\max}(\a,\b,\c)$ (together with its upper bound given in Eq.\ (\ref{minbound})) vs.\ the entanglement entropy of the borrowed state, $E(\c)$, for Example 4.}
    \end{minipage}
\end{figure}

As mentioned above, these examples show that there is not a universally optimal strategy for maximizing the gain in the minimal scenario. As dictated by our results in Sec.\ \ref{sec3b}, supercatalysis is not possible in the cases studied here when we borrow the most entangled catalyst. It is interesting to observe that otherwise the gain is always non-zero, except in the last example where it is also null for the least entangled catalyst. However, it is particularly noteworthy that $\tilde{G}_{max}(\a,\b)$ always attains modest values: below $1/4$ in Example 2 and around a meager $0.1$ in the other three. This feature is also observed in other examples and might lead to think that $\tilde{G}_{max}(\a,\b)$ is fundamentally limited. We address this question in the next subsection, which closes this article. 

\subsection{Minimal supercatalysis with entanglement gain arbitrarily close to 1}

The above-mentioned examples give rise to moderate values of the maximal gain in the minimal scenario and we have finished the previous subsection raising the question of whether there could be a more profound reason for this. Could it be that there is a universal constant $C < 1$ such that $\tilde{G}_{max}(\a,\b) \leq C$ holds for all non-trivial catalytic transformations $\ket{a} \rightarrow_c \ket{b}$, or at least for all those in the simplest scenario, where $SR(\a),SR(\b)\leq4$ and $C_2(\a,\b)\neq\emptyset$? In fact, as we prove in the following, a fundamental limitation does indeed hold in this latter case.

\begin{theorem}\label{th:nomaxgain}
Let $\ket{a},\ket{b}\in S(\mathcal{H}_1)$ be such that $SR(|a\rangle),SR(|b\rangle)\leq4$, $\ket{a}\nrightarrow\ket{b}$, and $C_2(\ket{a},\ket{b})\neq\emptyset$. Then, $\tilde{G}_{\max}(\a,\b)<1$.
\end{theorem}
\begin{proof}
The key observation here is that $\tilde{G}_{\max}(\a,\b)=1$ iff $\exists\c,\d\in S(\mathcal{H}_2)$ such that $G(\a,\b,\c,\d)=1$ (see\ Remark \ref{remark}), and, in turn, by the strict Schur-concavity of the entropy, this is equivalent to the transformation $\a\c\to\b\d$ being implementable by local unitary transformations, i.e.\ 
\begin{equation}\label{eq:nomaxgain}
(a\otimes c)^\downarrow=(b\otimes d)^\downarrow.
\end{equation}
Now, catalytic transformations are impossible when $SR(\a)<4$ and $SR(\b)<3$ \cite{JP}. Thus, it must hold that $SR(\a)=4$; but, if $SR(\c)=2$, Eq.\ (\ref{eq:nomaxgain}) cannot hold when $SR(\b)=3$. Therefore, we only need to consider the case $SR(\a)=SR(\b)=4$ and, therefore, $SR(\c)=SR(\d)=2$. Crucially, in this case it has been shown in \cite{multipartitecatalysis} that Eq.\ (\ref{eq:nomaxgain}) can only hold if $a=x\otimes y$ and $b=x\otimes z$ with $x,y,z\in\Delta_2$. However, this implies that either $y\succ z$ or $z\succ y$ and, consequently, that either $\a\to\b$ or $\b\to\a$, which is in contradiction with our assumption that $\ket{a}\nrightarrow\ket{b}$ and $\ket{a} \rightarrow_c \ket{b}$.
\end{proof}

The above theorem shows that, at least in the simplest setting, there is a fundamental difference between the optimal gain in the unrestricted scenario (see\ Proposition \ref{th:allsuper}) and the optimal gain in the minimal scenario. However, this fundamental restriction in the latter case only forbids maximal gain. It turns out that this is as much as one can restrict the gain in the minimal scenario from a universal perspective. Perhaps surprisingly given Theorem \ref{th:nomaxgain}, in the following we prove that, even within the case $SR(|a\rangle),SR(|b\rangle)\leq4$ and $C_2(\ket{a},\ket{b})\neq\emptyset$, one can always find input and output states $\a$ and $\b$ such that $\tilde{G}_{\max}(\a,\b)$ is as close to 1 as desired. Therefore, there exists no universal constant $C < 1$ such that $\tilde{G}_{max}(\a,\b) \leq C$ must hold $\forall\a,\b$.

\begin{theorem}\label{th:almostmaxgain}
For any $\delta>0$, $\exists\ket{a_\delta},\ket{b_\delta}\in S(\mathcal{H}_1)$ such that $SR(|a_\delta\rangle),SR(|b_\delta\rangle)\leq4$, $\ket{a_\delta}\nrightarrow\ket{b_\delta}$, and $C_2(\ket{a_\delta},\ket{b_\delta})\neq\emptyset$, for which it holds that $\tilde{G}_{\max}(\ket{a_\delta},\ket{b_\delta})>1-\delta$.
\end{theorem}
\begin{proof}
Let $\ket{a_\epsilon},\ket{b_\epsilon}\in S(\mathcal{H}_1)$ and $\ket{c_\epsilon},\ket{d_\epsilon}\in S(\mathcal{H}_2)$ be states with Schmidt vectors given by
\begin{equation}
a_\epsilon = (\frac{1}{2}, \frac{1}{2}-\epsilon, \frac{\epsilon}{2}, \frac{\epsilon}{2}), \quad b_\epsilon = (1-2\epsilon-\epsilon^2, \epsilon + \frac{\epsilon^2}{2}, \epsilon - \frac{\epsilon^2}{2}, \epsilon^2), \quad c_\epsilon = (\frac{1-2\epsilon-\epsilon^2}{1-\epsilon}, \frac{\epsilon+\epsilon^2}{1-\epsilon}),\quad d_\epsilon = (\frac{1}{2}+\sqrt{\epsilon}, \frac{1}{2}-\sqrt{\epsilon}),
\end{equation}
for $\epsilon\geq0$ sufficiently small so that all of the above are well-defined ordered Schmidt vectors. Now, using the continuity of the entanglement entropy, note that
\begin{equation}
\lim_{\epsilon\to0}G(\ket{a_\epsilon},\ket{b_\epsilon},\ket{c_\epsilon},\ket{d_\epsilon})=G(\ket{a_0},\ket{b_0},\ket{c_0},\ket{d_0})=1.
\end{equation}
As a consequence, for any $\delta>0$ we can find $\epsilon=\epsilon(\delta)>0$ sufficiently small as required above so that 
\begin{equation}
G(\ket{a_\epsilon},\ket{b_\epsilon},\ket{c_\epsilon},\ket{d_\epsilon})>1-\delta.
\end{equation} 
Thus, all is left is to prove is that these states give rise to a supercatalytic entanglement transformation. Namely, for all $\epsilon>0$ sufficiently small we have to show that: (i) $\ket{a_\epsilon}\nrightarrow\ket{b_\epsilon}$, (ii) $\ket{a_\epsilon}\ket{c_\epsilon}\rightarrow\ket{b_\epsilon}\ket{d_\epsilon}$, and (iii) $\ket{d_\epsilon}\to\ket{c_\epsilon}$ and $d_\epsilon\neq c_\epsilon$. In the following, we will no longer make the dependence of all objects on $\epsilon$ explicit in order to alleviate the notation.

Statements (i) and (iii) are relatively direct to prove. On the one hand, $f_2(a) = 1 - \epsilon$ and $f_2(b) = 1 - \epsilon - \epsilon^2/2$, so $f_2(a) > f_2(b)$ for all $\epsilon>0$. On the other hand, $SR(\d)=SR(\c)=2$ and $f_1(d)<f_1(c)$ holds for all $\epsilon>0$ sufficiently small. Thus, by Nielsen's theorem, we obtain the desired claims.

Let us then finish by proving (ii). By Taylor expanding at $\epsilon = 0$, the two coefficients of $c$ are given by
\begin{equation}
c_{1} = 1- \epsilon-2 \sum_{n=2}^{\infty}\epsilon^n,\quad c_{2}= \epsilon + 2\sum_{n=2}^{\infty}\epsilon^n.
\end{equation}
Therefore, for all $\epsilon>0$ small enough we have that
\begin{align}
&(a\otimes c)^\downarrow_1= a_1c_{1} = \frac{1}{2} - \frac{\epsilon}{2} - \sum_{n=2}^{\infty}\epsilon^n,\quad
(a\otimes c)^\downarrow_2= a_2 c_{1} = \frac{1}{2} - \frac{3}{2}\epsilon + \sum_{n=3}^{\infty}\epsilon^n,\quad 
(a\otimes c)^\downarrow_3= a_1  c_{2} =  \frac{\epsilon}{2} + \sum_{n=2}^{\infty}\epsilon^n,  \notag \\ 
&(a\otimes c)^\downarrow_4= a_2  c_{2} = \frac{\epsilon}{2} -  \sum_{n=3}^{\infty}\epsilon^n,\quad 
(a\otimes c)^\downarrow_5= a_3  c_{1} =  \frac{\epsilon}{2} - \frac{\epsilon^2}{2} - \sum_{n=3}^{\infty}\epsilon^n,\quad
(a\otimes c)^\downarrow_6= a_4  c_{1} = a_3  c_{1},  \notag \\ 
&(a\otimes c)^\downarrow_7= a_3  c_{2} =  \frac{\epsilon^2}{2} + \sum_{n=3}^{\infty}\epsilon^n, \quad
(a\otimes c)^\downarrow_8= a_4  c_{2} =  a_3  c_{2},
\end{align}
and that
\begin{align}
&(b\otimes d)^\downarrow_1=b_1  d_1 = \frac{1}{2} + \sqrt{\epsilon} - \epsilon - 2 \epsilon \sqrt{\epsilon} - \frac{\epsilon^2}{2} - \epsilon^2 \sqrt{\epsilon}, \quad 
(b\otimes d)^\downarrow_2=b_1  d_2 =  \frac{1}{2} - \sqrt{\epsilon} - \epsilon + 2 \epsilon \sqrt{\epsilon} - \frac{\epsilon^2}{2} + \epsilon^2 \sqrt{\epsilon},  \notag \\  
&(b\otimes d)^\downarrow_3=b_2  d_1 = \frac{\epsilon}{2} + \epsilon \sqrt{\epsilon} + \frac{\epsilon^2}{4} + \frac{\epsilon^2}{2}\sqrt{\epsilon},   \quad 
(b\otimes d)^\downarrow_4=b_3  d_1 =  \frac{\epsilon}{2} + \epsilon \sqrt{\epsilon}  - \frac{\epsilon^2}{4} - \frac{\epsilon^2}{2}\sqrt{\epsilon}, \notag \\ 
&(b\otimes d)^\downarrow_5=b_2  d_2 =\frac{\epsilon}{2} - \epsilon \sqrt{\epsilon}  + \frac{\epsilon^2}{4} - \frac{\epsilon^2}{2}\sqrt{\epsilon}, \quad 
(b\otimes d)^\downarrow_6=b_3  d_2 =  \frac{\epsilon}{2} - \epsilon \sqrt{\epsilon} - \frac{\epsilon^2}{4} + \frac{\epsilon^2}{2}\sqrt{\epsilon}, \notag \\ 
&(b\otimes d)^\downarrow_7=b_4  d_1 = \frac{\epsilon^2}{2} + \epsilon^2\sqrt{\epsilon},  \quad
(b\otimes d)^\downarrow_8=b_4  d_2 = \frac{\epsilon^2}{2} - \epsilon^2\sqrt{\epsilon}.
\end{align}
Hence, 
\begin{align}
& f_1(b\otimes d)-f_1(a\otimes c)=\sqrt{\epsilon}+o(\sqrt{\epsilon}), \quad f_2(b\otimes d)-f_2(a\otimes c)=0, \quad f_3(b\otimes d)-f_3(a\otimes c)=\epsilon\sqrt{\epsilon}+o(\epsilon\sqrt{\epsilon}), \notag\\
& f_4(b\otimes d)-f_4(a\otimes c)=2\epsilon\sqrt{\epsilon}+o(\epsilon\sqrt{\epsilon}),\quad f_5(b\otimes d)-f_5(a\otimes c)=\epsilon\sqrt{\epsilon}+o(\epsilon\sqrt{\epsilon}),\notag \\ 
& f_6(b\otimes d)-f_6(a\otimes c)=2\sum_{n=3}^{\infty}\epsilon^n,\quad f_7(b\otimes d)-f_7(a\otimes c)=\epsilon^2\sqrt{\epsilon}+o(\epsilon^2\sqrt{\epsilon}).
\end{align}
Thus, $b\otimes d\succ a\otimes c$ holds for all $\epsilon>0$ sufficiently small and, using Nielsen's theorem, this proves (ii).
\end{proof}

It might be illuminating to point out that, by Theorem \ref{th:catalyst}, the catalysts for the states $\a$ and $\b$ used in this proof satisfy, for $\epsilon>0$ sufficiently small,
\begin{equation} 
x_{\min}(a,b)=\frac{1+\epsilon}{2},\quad  x_{\max}(a,b)=\frac{1-2\epsilon - \epsilon^2}{1- \epsilon}.
\end{equation}
Thus, this transformation has a large range of catalysts: any state between a separable state and the maximally entangled state is a valid catalyst as $\epsilon\to0$. The borrowed state $\c$ in the proof is precisely the least entangled catalyst of Schmidt rank equal to 2 for this transformation. On the other hand, the returned state $\d$ is not the most entangled catalyst of Schmidt rank equal to 2. This is because the supercatalytic transformation happens not to be possible with such choice for $\d$. Still, the supercatalytic transformation is possible by taking a returned state that also approaches the maximally entangled state as $\epsilon\to0$ but only at a slower speed ($O(\sqrt{\epsilon})$ instead of $O(\epsilon)$). This construction illustrates the power of supercatalysis. The catalytic transformation incurs in a severe loss of entanglement as $\epsilon$ becomes small: although impossible without the aid of an auxiliary system, it takes a state close to a maximally entangled Schmidt-rank-2 state into a state that is almost separable. However, in the supercatalytic scenario, following the miserly strategy we recover most of this loss in the auxiliary system by lending a 2-qubit catalyst with almost zero entanglement and obtaining in return a 2-qubit state with almost maximal entanglement. Lastly, note that this construction does not attain $\tilde{G}_{max}(\a,\b) = 1$ because that would happen exactly at $\epsilon=0$, where it no longer holds that $\ket{a}\nrightarrow\ket{b}$.

\section{Conclusions}\label{sec5}

In this work we have introduced the entanglement gain as a figure of merit taking values in $[0,1]$ that measures the quality of supercatalytic entanglement transformations. While we have shown that the problem trivializes in the case where Scrooge has complete freedom in choosing the state he lends, in the sense that all catalytic transformations are supercatalyzable with maximal gain, the situation is much richer when this choice is constrained. In fact, we have provided several results that show that the performance of supercatalytic protocols drastically depends on the choice of borrowed state. First, we have proved that a large class of catalytic transformations are not fully supercatalyzable, i.e.\ there is always an appalling strategy in which the entanglement gain is bound to be zero. Second, we have studied in more detail the particularly relevant case of minimal supercatalytic transformations, where the Schmidt rank of the states Scrooge can lend is constrained to the minimum possible, and we have concluded that there exist instances of catalytic transformations that are not minimally supercatalyzable. Furthermore, we have shown that in certain scenarios the gain for minimal supercatalytic transformations can never be 1. This notwithstanding, we have proved that in these same scenarios there is a family of minimal supercatalytic transformations with gain as close to 1 as desired, which shows that, despite the above observation, there is no universal bound for the optimal gain in minimal supercatalytic transformations. In addition to this, we have performed an exhaustive study in the simplest scenario and we have provided a state-dependent upper bound on the gain. We have found that a miserly strategy is often optimal and that our bound can be tight (meaning that the least entangled catalyst can be lent in exchange for the maximally entangled catalyst). However, this is not always true and, actually, we have identified examples where an intermediate strategy greatly outperforms the miserly strategy.  

Supercatalytic transformations offer a clear enhancement over merely catalytic ones. However, while supercatalysis was introduced more than 20 years ago in \cite{BandRoy}, it has remained largely unexplored. We hope that the present article stimulates further research in this area. In particular, given that catalysis has many applications and provides an advantage in several protocols, it would be interesting to find specific information-processing tasks where this advantage is boosted by supercatalytic strategies. From the technical side, our work leads to several interesting open questions. Is there an example of a fully supercatalyzable transformation or are these impossible in general? Our impossibility proofs rely on the fact that under certain hypotheses one cannot borrow a most entangled catalyst of Schmidt rank less than or equal to $r$ and perform a supercatalytic transformation returning a state of Schmidt rank larger than $r$. Is this ever possible? On a different note, can we find tighter state-dependent upper bounds on $G_{max}(\a,\b,\c)$ and $\tilde{G}_{max}(\a,\b)$ than the one given in Theorem \ref{th:bound}, even if only in the simplest scenario (see\ Eqs.\ (\ref{minbound}) and (\ref{optminbound}))? Also, even though the miserly strategy seems to perform well in general, we have found instances where it leads to zero gain. Can one find a general strategy that is guaranteed to perform well always (i.e.\ with a desired lower bound on the gain)?

\begin{acknowledgments}
We acknowledge financial support from the Spanish Ministerio de Ciencia, Innovaci\'on y Universidades (grant PID2023-146758NB-I00 and grant PID2024-160539NB-I00 funded by MCIN/AEI/10.13039/501100011033). J. I. de V. also acknowledges financial support from the Spanish Ministerio de Ciencia, Innovaci\'on y Universidades (``Severo Ochoa Programme for Centres of Excellence'' grant CEX2023-001347-S funded by MCIN/AEI/10.13039/501100011033) and from Comunidad de Madrid (grant QUITEMAD-CM TEC-2024/COM-84).
\end{acknowledgments}

\bibliography{supercatbib}

\end{document}